\documentclass[12pt]{article}
\usepackage{amsfonts}
\usepackage{amssymb}
\usepackage{amsmath}
\usepackage{color}
\usepackage[unicode,bookmarks,bookmarksopen,bookmarksopenlevel=2,colorlinks,linkcolor=blue,citecolor=green]{hyperref}

\newcommand*{\pd}
[2]{\mathchoice{\frac{\partial#1}{\partial#2}}
  {\partial#1/\partial#2}{\partial#1/\partial#2}
  {\partial#1/\partial#2}}

\newcommand{\ddx}[1]{\partial_x^{#1}}
\newtheorem{theorem}{Theorem}

\newtheorem{lemma}[theorem]{Lemma}

\newtheorem{remark}[theorem]{Remark}

\newenvironment{proof}[1][Proof]{\textbf{#1.} }{\ \rule{0.5em}{0.5em}}
\allowdisplaybreaks[3]
\input{tcilatex}

\title{\bfseries Bi-Hamiltonian structure of the \\ Oriented Associativity
Equation}
\author{Maxim~V.~Pavlov$^{1}$, Raffaele~F.~Vitolo$^{2}$ \\
$^{1}$Lebedev Physical Institute of Russian Academy of Sciences,\\
Moscow, Leninskij Prospekt, 53;\\
$^{2}$Department of Mathematics and Physics,\\
University of Salento, Lecce, Italy\\
and INFN, Section of Lecce\\
\url{http://poincare.unisalento.it/vitolo}}
\date{}

\begin{document}

\maketitle

\begin{abstract}
  The Oriented Associativity equation plays a fundamental role in the theory of
  Integrable Systems. In this paper we prove that the equation, besides being
  Hamiltonian with respect to a first-order Hamiltonian operator, has a
  third-order non-local homogeneous Hamiltonian operator belonging to a class
  which has been recently studied, thus providing a highly non-trivial example
  in that class and showing intriguing connections with algebraic geometry.

\end{abstract}

\section{Introduction}

The Associativity equation, or Witten-Dijkgraaf-Verlinde-Verlinde (WDVV)
equation, plays a fundamental role in the geometric theory of Integrable
Systems. Its solutions define Frobenius manifolds, which correspond to
integrable systems; Frobenius manifolds also play a fundamental role in the
theory of quantum cohomology and Gromov--Witten invariants. These connections
were shown by B. Dubrovin in his seminal paper \cite{Dubr}.

The nonlinear partial differential system of equations
\begin{equation}
\frac{\partial ^{2}c^{i}}{\partial a^{j}\partial a^{m}}\frac{\partial
^{2}c^{m}}{\partial a^{k}\partial a^{n}}=\frac{\partial ^{2}c^{i}}{\partial
a^{k}\partial a^{m}}\frac{\partial ^{2}c^{m}}{\partial a^{j}\partial a^{n}}
\label{ori}
\end{equation}%
on $N$ unknown functions $(c^i)$ of $N$ independent variables $(a^j)$ was
introduced in \cite{HM} as a generalization of the Associativity equations.
Its solution define $F$-manifolds, which are still in correspondence with
integrable systems. The far-reaching implication of this generalization are an
active subject of study: flat and bi-flat $F$-manifolds have interesting
connections with Painlev\'e equations \cite{AL16,LPR11,Lor14}; see also the
papers \cite{KoMa07b,KoOr09}) devoted to coisotropic deformations. We call the
system~\eqref{ori} the \textit{Oriented Associativity equation}.

The Oriented Associativity equation admits the scalar linear spectral problem
\begin{equation}
\frac{\partial ^{2}h}{\partial a^{i}\partial a^{j}}=\lambda \frac{\partial
^{2}c^{m}}{\partial a^{i}\partial a^{j}}\frac{\partial h}{\partial a^{m}}
\label{zak}
\end{equation}%
(see, for instance, \cite{Dubr,SPori})%
that ensure that the equation is integrable as it provides a Lax
pair.

We observe that the Associativity equation \cite{Dubr} can be obtained from
(\ref{ori}) by the potential reduction
$c^{i}=\eta ^{im}\pd{F}{a^{m}}$, where $\eta^{ks}$ is a
constant nondegenerate symmetric matrix.

In this paper, we will \emph{prove the existence of a bi-Hamiltonian formalism
  for the Oriented Associativity equation \eqref{ori}}.

The above result has strong analogies with the known results on the
Associativity (WDVV) equation. Indeed, the Associativity equation can be
written as $N-2$ commuting hydrodynamic-type systems of conservation laws
\cite{FM96}. For $N=3$ the (only) system was shown to be bi-Hamiltonian in
\cite{FGMN97}.  Further investigations shown a similar situation in the case
$N=4$: the Hamiltonian operators were, in both cases, a first-order homogeneous
operator (found in \cite{FM96} for $N=4$) and a third-order homogeneous
operator (found in \cite{PV15} for $N=4$). First-order homogeneous operators
\cite{DN83} can be written as $A_1^{ij}=g^{ij}\partial_x$, where $(g^{ij})$ is
a constant matrix, in a suitable coordinate system; third-order homogeneous
operators \cite{DN84} have a more complicated structure, and can be brought to
the form
\begin{equation}
  \label{eq:14}
  A_2^{ij}=\ddx{}\left(g^{ij}\ddx{}+c_{k}^{ij}u_{x}^{k}\right)\ddx{}.
\end{equation}
Such operators have been extensively studied quite recently \cite{FPV14,FPV16}.

Recently, a first-order homogeneous Hamiltonian operator for the simplest case
($N=3$) of Oriented Associativity equation (written in the form of a
hydrodynamic-type system) was found \cite{SPori} with the same method as in the
Associativity case.  It was natural to conjecture that a third-order
homogeneous Hamiltonian operator might exist.

In this paper, we will \emph{prove that the Oriented Associativity equation in
  hydrodynamic-type form admits a \textbf{non-local} third-order homogeneous
  Hamiltonian operator} of a class that was recently introduced by M. Casati,
E.V. Ferapontov and the authors of this paper in \cite{CFPV18}.

The significance of the result is high: indeed, it is known that the
Associativity equations (in hydrodynamic form) in the cases $N=3$ and $N=4$
discussed above correspond to linear line congruences, which are algebraic
varieties in the Pl\"ucker variety all lines of a projective space
\cite{FPV17:_system_cl}. Their third-order Hamiltonian operators correspond to
quadratic line complexes, which are algebraic varieties of lines in a
projective space of different dimension with respect to the previous lines
\cite{FPV14,FPV16}.

The Oriented Associativity equation (in hydrodynamic form) can also be
interpreted as a line congruence, even if we still do not know if the
congruence is linear.  The third-order non-local homogeneous Hamiltonian
operator that we find as the main result in this paper also defines a quadratic
line complex.  It is thus clear that the strong links between the Associativity
equation and projective-geometric varieties are preserved for the Oriented
Associativity equation. We believe that such structures play an important role
in the rich geometry of such equations, with lots of interesting Mathematics
yet to be discovered.

The computation related to finding the non-local Hamiltonian operator is highly
non-trivial, and it is made possible by a systematic use of computer algebra
systems, in particular Reduce and its package \texttt{CDE} for computations
with Hamiltonian operators \cite{KVV17,VitoloHO}.

\section{The Oriented Associativity equation}

The system of quadratic equations 
\begin{equation}
  \label{eq:11}
\begin{split}
&u_{xx}=v_{xt}w_{xx}-v_{xx}w_{xt}+w_{xt}^{2}-w_{xx}w_{tt}, \\
&u_{xt}=v_{tt}w_{xx}-v_{xt}w_{xt}, \\
&u_{tt}=v_{xt}^{2}-v_{xx}v_{tt}+v_{tt}w_{xt}-v_{xt}w_{tt},
\end{split}%
\end{equation}
is the Oriented Associativity equation in the simplest case $N=3$. It is
endowed by the Lax pair
\begin{align}
  \label{eq:24}
  &\begin{pmatrix}
    \psi \\ \psi_1 \\ \psi_2
  \end{pmatrix}_x = \lambda
  \begin{pmatrix}
    0 & 1 & 0 \\ u_{xx} & v_{xx} & w_{xx} \\ u_{xt} & v_{xt} & w_{xt}
  \end{pmatrix}
  \begin{pmatrix}
    \psi \\ \psi_1 \\ \psi_2
  \end{pmatrix},
  \\ \label{eq:25}
  &
  \begin{pmatrix}
    \psi \\ \psi_1 \\ \psi_2
  \end{pmatrix}_t = \lambda
  \begin{pmatrix}
    0 & 0 & 1 \\ u_{xt} & v_{xt} & w_{xt} \\ u_{tt} & v_{tt} & w_{tt}
  \end{pmatrix}
  \begin{pmatrix}
    \psi \\ \psi_1 \\ \psi_2
  \end{pmatrix}
\end{align}
Let us introduce a new set of field variables $%
q^{1}=u_{xx},q^{2}=u_{xt},q^{3}=v_{xx},q^{4}=v_{xt}$, $%
q^{5}=w_{xx},q^{6}=w_{xt}$. Then, the quadratic system~\eqref{eq:11} becomes
the six component hydrodynamic type system of conservation laws
\begin{eqnarray}
q_{t}^{1} &=&q_{x}^{2},\text{ \ \ \ \ \ \ \ \ }q_{t}^{2}=\partial _{x}\frac{%
q^{2}q^{6}+q^{1}q^{4}-q^{2}q^{3}}{q^{5}},  \notag \\
&&  \notag \\
q_{t}^{3} &=&q_{x}^{4},\text{ \ \ \ \ \ \ \ \ }q_{t}^{4}=\partial _{x}\frac{%
q^{2}+q^{4}q^{6}}{q^{5}},  \label{six} \\
&&  \notag \\
q_{t}^{5} &=&q_{x}^{6},\text{ \ \ \ \ \ \ \ \ }q_{t}^{6}=\partial _{x}\frac{%
(q^{6})^{2}-q^{3}q^{6}+q^{4}q^{5}-q^{1}}{q^{5}}.  \notag
\end{eqnarray}%
The above system (\ref{six}) possesses at least two extra local
conservation laws. One can prove that this system has also three extra
conservation laws%
\begin{equation}\label{eq:17}
\partial _{t}v^{k}=\partial _{x}\frac{(v^{k})^{2}-q^{3}v^{k}-q^{1}}{q^{5}},%
\text{ \ }k=1,2,3,
\end{equation}%
where $v^i$ are the roots of the characteristic polynomial $\lambda^{3}-(q^{3}+q^{6})\lambda^{2}+ (q^{3}q^{6}-q^{4}q^{5}-q^{1})\lambda +
q^{1}q^{6}-q^{2}q^{5}$ of one of the
matrices of the Lax pair~\eqref{eq:24}. By Vi\`{e}te's theorem we have
$q^{3}+q^{6}=v^{1}+v^{2}+v^{3}$, so that only two of the densities $v^i$ are new.

\section{First-order Hamiltonian structure}

The hydrodynamic-type system \eqref{six} admits a first-order homogeneous
Hamiltonian operator. This class of Hamiltonian operators was first introduced
in \cite{DN83}. Operators in this class always admit a coordinate system in
which they can be presented as $A_1 = g^{ij}\partial_x$, where $g^{ij}$ is a
constant matrix. The results in this sections were found in~\cite{SPori}, using
techniques that are analogous to those used in~\cite{FGMN97}.

We can change the coordinates in the above hydrodynamic-type system to the
new coordinates $(u^k)$ defined by the Vi\`ete formulae: 
\begin{equation}
\begin{split}
&u^{1}=v^{1},\qquad u^{2}=v^{2}, \\
&u^{3}=v^{3},\qquad u^{4}=q^{4}, \\
&u^{5}=q^{5},\qquad u^{6}=2q^{3}-(v^{1}+v^{2}+v^{3}).
\end{split}
\label{eq:12}
\end{equation}
They are related with $(q^i)$ by the formulae
\begin{equation}  \label{eq:13}
\begin{split}
& q^{1}=\frac{1}{4}
(u^{1}+u^{2}+u^{3})^{2}-(u^{1}u^{2}+u^{1}u^{3}+u^{2}u^{3})-\frac{1}{4}%
(u^{6})^{2}-u^{4}u^{5}, \\
&q^{2}=\frac{2u^{1}u^{2}u^{3}+(u^{1}+u^{2}+u^{3}-u^{6})q^{1}}{2u^{5}}, \\
&q^{3}=\frac{1}{2}(u^{1}+u^{2}+u^{3}+u^{6}), \\
&q^{4}=u^{4}, \\
&q^{5}=u^{5}, \\
&q^{6}=\frac{1}{2}(u^{1}+u^{2}+u^{3}-u^{6}).
\end{split}%
\end{equation}
Note that the inverse formulae contain cubic roots, and have a much more
complicated expression. In the coordinates $(u^i)$ a Hamiltonian formulation of
the system becomes immediate:
\begin{equation}\label{eq:26}
u_{t}^{i}=g^{ik}\partial _{x}\frac{\partial H}{\partial u^{k}},
\quad\text{where}\quad
g^{ik}=-
\begin{pmatrix}
0 & 1 & 1 & 0 & 0 & 0 \\ 
1 & 0 & 1 & 0 & 0 & 0 \\ 
1 & 1 & 0 & 0 & 0 & 0 \\ 
0 & 0 & 0 & 0 & 1 & 0 \\ 
0 & 0 & 0 & 1 & 0 & 0 \\ 
0 & 0 & 0 & 0 & 0 & 2%
\end{pmatrix},
\end{equation}
the Hamiltonian density is $H=q^2$ and the momentum density is
$P=q^{1}=\frac{1}{2}\tilde{g}_{ik}u^{i}u^{k}$.

\section{Third-order nonlocal operators and\newline
systems of conservation laws}

After the results in \cite{SPori}, we might be tempted to conjecture that, by
analogy with the Associativity equation, also the Oriented Associativity
equation is endowed by two local homogeneous Hamiltonian operators, of first
order and third order. Strictly speaking, this is not true.
We recall that the conditions for an operator of the form~\eqref{eq:14} to be
Hamiltonian (provided $\det(g^{ij})\neq 0$) are
\begin{subequations}\label{eq:32}
\begin{align}
  \label{eq:22}
  &g_{ij}=g_{ji}, \\
  \label{eq:28}
  &c_{nkm}=\frac{1}{3}(g_{mn,k}-g_{kn,m})\\
  \label{eq:23}
  &g_{ij, k}+g_{jk, i}+g_{ki, j}=0, \\
  \label{eq:27}
  &c_{nml,k}+c^s_{ml}c_{snk} =0.
\end{align}
\end{subequations}
where $(g_{ij})^{-1}=(g^{ij})$ and $c_{ijk}=g_{iq}g_{jp}c_{k}^{pq}$.
By repeating the procedure that led to the results in \cite{PV15}, we found a
candidate $(g^{ij})$ for a leading term of a third order homogeneous
Hamiltonian operator. But that fulfills~\eqref{eq:22}, \eqref{eq:28} and
\eqref{eq:23} but does not fulfill~\eqref{eq:27}. After the results
in~\cite{CFPV18}, we conjectured that the third-order homogeneous Hamiltonian
operator $B=(B^{ij})$ might be nonlocal, of the type
\begin{equation}  \label{casimir-nl}
  B^{ij}=\partial_x\circ F^{ij}\circ \partial_x =
  \partial_x^{}(g^{ij}\partial_x^{} + c^{ij}_ku^k_x + c^\alpha
  w^i_{\alpha k}u^k_x\partial_x^{-1}w^j_{\alpha h}u^h_x)\partial_x^{}
\end{equation}
and $w^i_{\alpha k} = w^i_{\alpha k}(u^j)$, with $c^\alpha\in\mathbb{R}$.  In
such an ansatz $F^{ij}$ has the same structure as Ferapontov--Mokhov nonlocal
first-order homogeneous operators
\cite{FM90:_nonloc_hamil,Fer+curv,FerFirst}. However, the two compositions with
$\partial_x$ change the conditions for the operator to be Hamiltonian to
\cite{CFPV18}
\begin{subequations}\label{eq:35}
  \begin{align}
    \label{eq:29}
    &w_{\alpha ij}+w_{\alpha ji}=0, \\
    \label{eq:30}
    &w_{\alpha ij, l}-c^s_{ij}w_{\alpha sl}=0, \\
    \label{eq:31}
    &c_{nml,k}+c^s_{ml}c_{snk} + c^\alpha w_{\alpha ml}w_{\alpha nk}=0,
  \end{align}
\end{subequations}
in addition to~\eqref{eq:22}, \eqref{eq:28}, \eqref{eq:23} (of
course,~\eqref{eq:31} is a modification of ~\eqref{eq:27}), where
$w_{ij}=g_{is}w^s_j$.  We remain with the problem of determining the tensors
$w^i_{\alpha j}$. It is known that in Ferapontov--Mokhov case they are matrices
of commuting flows with respect to the hydrodynamic-type system of which the
operator is Hamiltonian.  In this case, that is false: the condition of
compatibility between $B$ and the Oriented Associativity equation~\eqref{six}
can be derived by the condition that the Hamiltonian operator maps conserved
quantities into symmetries. It was shown in
\cite{KerstenKrasilshchikVerbovetsky:HOpC} that such a condition is equivalent
to finding solutions $B$ to the equation
\begin{equation}
  \label{preHam} \ell _{E}(B(\mathbf{p}))=0,
\end{equation}
over the \emph{adjoint system} (or cotangent covering)
\begin{equation}
\left\{ 
\begin{array}{l}
  E=0, \\ 
  \ell _{E}^{\ast }(\mathbf{p})=0,
\end{array}%
\right.  \label{cot}
\end{equation}%
where $\mathbf{p}$ is an auxiliary (vector) variable, $E^i=u^i_t - (V^i)_x =
0$ is the initial equation, with $V^i=V^i(\mathbf{u})$ the vector of fluxes,
and $\ell_E$ is the formal linearization (Fr\'echet derivative) of $E$ and
$\ell^*_E$ its adjoint operator.

It is easier to compute the condition~\eqref{preHam} in potential coordinates
$b^i_x=u^i$. We have
\begin{equation}  \label{casimir-nl-pot}
  B^{ij}= - (g^{ij}(\mathbf{b}_x)\partial_x^{} + c^{ij}_k(\mathbf{b}_x)b^k_{xx}
  + c^\alpha
  w^i_{\alpha k}(\mathbf{b}_x)b^k_{xx}
  \partial_x^{-1}
  w^j_{\alpha h}(\mathbf{b}_x)b^h_{xx})
\end{equation}
and $E^i = b^i_t - V^i(\mathbf{b}_x)$. Let us introduce the notation
\begin{displaymath}
  \frac{\partial V^i}{\partial b_x^k}= V^i_k,\quad
  \frac{\partial V^i}{\partial b_x^k\partial b^h_x}= V^i_{kh},\quad
  g^{ij}_{,k} = \frac{\partial g^{ij}}{\partial b^k_x},\quad
  c^{ij}_{k,h} = \frac{\partial c^{ij}_k}{\partial b^h_x},
\end{displaymath}
and similarly for other derivatives. We have
\begin{equation}  \label{eq:251}
\ell_F(\boldsymbol{\varphi}) = \partial_t\varphi^i -
V^i_j\partial_x\varphi^j,\quad 
\ell_F^*(\boldsymbol{\psi}) = -\partial_t\psi_k + \partial_x(V^i_k\psi_i).
\end{equation}
so that the adjoint system is
\begin{align}  \label{eq:271}
  &b_{t}^{i}=V^{i}(\mathbf{b}_{x}) \\
  \label{eq:39}
  &p_{k,t} = V^i_{kh}b^h_{xx}p_i + V^i_{k}p_{i,x}
\end{align}
If we assume that $w^i_{\alpha j}(\mathbf{b}_x)b^j_{xx}$ are hydrodynamic-type
symmetries of the system~\eqref{eq:271} then we can prove that they yield
conservation laws on the adjoint system whose densities and fluxes are,
respectively:
\begin{equation}
  \label{eq:40}
  r_{\alpha t} = V^i_j w^j_{\alpha k} b^k_{xx}p_i,
  \quad
  r_{\alpha x} = w^i_{\alpha k} b^k_{xx} p_i.
\end{equation}
The potential variables $r_\alpha$ allow us to represent the operator as
\begin{equation}  \label{eq:451}
B^i(\mathbf{p}) = -g^{ij}p_{j,x} - c^{ij}_kb^k_{xx}p_j - c^\alpha
w^i_{\alpha k}b^k_{xx}r_\alpha.
\end{equation}
\begin{lemma}
  The condition $\ell_F(B(\mathbf{p})) = 0$ is equivalent to the
  conditions:
  \begin{subequations}
\label{eq:467}
\begin{align}  \label{eq:33}
  &- g^{ij} V^h_j + V^i_j g^{jh} = 0 \\
  \label{eq:44}
& - g^{ih}_{,k} V^k_l - g^{ij}2V^h_{jl} - c^{ij}_l V^h_j + V^i_jg^{jh}_{,l}
                  + V^i_jc^{jh}_l=0 \\
  \label{eq:45}
& -g^{ik} V^h_{kl} - c^{ih}_k V^k_l + V^i_kc^{kh}_l = 0. \\
\begin{split}  \label{eq:21}
& - g^{ij}V^h_{jlm} - \frac{1}{2}\big( c^{ih}_{m,j} V^j_l + c^{ih}_{l,j}
V^j_m \big) - c^{ih}_k V^k_{lm} \\
&\hphantom{+} - \frac{1}{2}\big( c^{ij}_m V^h_{jl} + c^{ij}_l V^h_{jm} \big) %
+\frac{1}{2}\big(V^i_jc^{jh}_{m,l} + V^i_j c^{jh}_{l,m}\big) \\
&\hphantom{+} - c^\alpha\frac{1}{2} \big( V^h_k(w^i_{\alpha l}w^k_{\alpha m}
+ w^i_{\alpha m}w^k_{\alpha l}) + V^i_j(w^j_{\alpha l}w^h_{\alpha m}+
w^j_{\alpha m}w^h_{\alpha l})\big) = 0
\end{split}
\\
\begin{split}\label{eq:42}
& - w^i_{\alpha h,k}V^k_m - w^i_{\alpha m,k}V^k_h - w^i_{\alpha k}V^k_{m,h}
\\
&\hphantom{+}- w^i_{\alpha k}V^k_{h,m} + V^i_k w^k_{\alpha m,h} + V^i_k
w^k_{\alpha h,m} = 0
\end{split}
  \\
  \label{eq:43}
& -w^i_{\alpha k} V^k_h + V^i_k w^k_{\alpha h} =0
\end{align}
\end{subequations}
\end{lemma}
\begin{proof}
We have: 
\begin{align*}
  \ell_F(&B(\mathbf{p}))^i = \\
= & - g^{ij}_{,k} b^k_{xt}p_{j,x} - g^{ij}p_{j,xt} - c^{ij}_{k,h}
b^h_{xt}b^k_{xx}p_j - c^{ij}_kb^k_{xxt}p_j - c^{ij}_kb^k_{xx}p_{j,t} \\
& + V^i_j \big( g^{jk}_{,h}b^h_{xx}p_{k,x} + g^{jk}p_{k,xx} +
c^{jh}_{k,l}b^l_{xx}b^k_{xx}p_h +c^{jh}_kb^k_{xxx}p_h +
c^{jh}_kb^k_{xx}p_{h,x} \big) \\
& - c^\alpha w^i_{\alpha k,h}b^h_{xt}b^k_{xx}r_\alpha - c^\alpha w^i_{\alpha
k}b^k_{xxt}r_\alpha - c^\alpha w^i_{\alpha k}b^k_{xx}V^h_lw^l_{\alpha m}
b^m_{xx}p_h \\
& + V^i_jc^\alpha \big(w^j_{\alpha k,h}b^h_{xx}b^k_{xx}r_\alpha +
w^j_{\alpha k}b^k_{xxx}r_\alpha + w^j_{\alpha k}b^k_{xx}w^h_{\alpha
l}b^l_{xx}p_h).
\end{align*}
After replacing the derivatives $b^h_{xt}$, $b^h_{xxt}$ and $p_{j,xt}$ using
the equations~\eqref{eq:271} and \eqref{eq:39} we obtain a polynomial in $p_j$,
$b^k_{xx}$ and higher $x$-derivatives; its coefficient shall vanish, they are
are the conditions of the statement.
\end{proof}

\begin{remark}
  A direct computation shows that the two flows $V^i$ and
  $w^j_{\alpha k}u^k_{xx}$ commute if and only if the conditions \eqref{eq:42}
  and \eqref{eq:43} hold true. Moreover, by arguments that are similar to those
  of \cite[Theorem~1]{FPV17:_system_cl} it can be proved that~\eqref{eq:21} is
  a consequence of the other equations~\eqref{eq:467} and~\eqref{eq:32},
  \eqref{eq:35}.
\end{remark}

\section{Third-order Hamiltonian structure}

It is known\cite{FPV14} that $g_{ij}$ shall be a polynomial of second degree.
Then, equations~\eqref{eq:33}, \eqref{eq:44}, \eqref{eq:45} are easily solved
with respect to $(g_{ij})$, where $V^i$ is the vector of fluxes of the Oriented
Associativity equation \eqref{six}. We obtain the unique solution:
\begin{multline}  \label{eq:10}
{\scriptsize (g_{ij}) = } \\
{\scriptsize \left(%
\begin{matrix}
2 & 0 & q^3 & -q^5 \\ 
0 & 0 & 2q^5 & 0 \\ 
q^3 & 2q^5 & - 2(q^1 + q^4q^5) & q^5(q^3 - q^6) \\ 
-q^5 & 0 & q^5(q^3 - q^6) & 2(q^5)^2 \\ 
2q^4 & - q^3 + 2q^6 & - q^2 + q^4(q^3 - q^6) & - q^1 + q^3(- q^3 + 2q^6) -
2q^4q^5 - (q^6)^2 \\ 
0 & -q^5 & 2q^4q^5 & q^5( - q^3 + q^6)%
\end{matrix}%
\right.} \\
{\scriptsize \left. 
\begin{matrix}
2q^4 & 0 \\ 
-q^3+2q^6 & -q^5 \\ 
- q^2 + q^4(q^3 - q^6) & 2q^4q^5 \\ 
- q^1 + q^3(- q^3 + 2q^6) - 2q^4q^5 - (q^6)^2 & q^5( - q^3 + q^6) \\ 
2(q^4)^2 & - q^2 +q^4(- q^3 + q^6) \\ 
- q^2 +q^4( - q^3 + q^6) & - 2q^4q^5%
\end{matrix}%
\right)}
\end{multline}
We stress that $(g_{ij})$ does not fulfill \eqref{eq:27}. Hence, we shall
compute suitable tensors $w^i_{\alpha j}$.

A direct computation of $w^i_{\alpha j}$ as symmetries of \eqref{eq:271} is
very heavy. Since we have at our disposal a Lax pair, we can compute a sequence
of homogeneous conserved quantities with a standard technique in the theory of
integrable systems; see \cite{PV15} for details. Then, we can
transform them into symmetries using the Hamiltonian operator. We rewrite
\eqref{eq:24} in terms of $(q^i)$ and get
\begin{equation}  \label{eq:139}
\begin{pmatrix}
\psi \\ 
\psi^1 \\ 
\psi^2%
\end{pmatrix}%
_x = \lambda%
\begin{pmatrix}
0 & 1 & 0 \\ 
q^1 & q^3 & q^5 \\ 
q^2 & q^4 & q^6%
\end{pmatrix}
\begin{pmatrix}
\psi \\ 
\psi^1 \\ 
\psi^2%
\end{pmatrix}%
.
\end{equation}
By eliminating $\psi^{1}$, $\psi^{2}$ from \eqref{eq:139} we obtain the single
linear PDE 
\begin{multline*}
\big( - q^{1}_x\lambda^{2}q^{5} + q^{5}_x\lambda^{2}q^{1} + \lambda^{3}
q^{1}q^{5}q^{6} - \lambda^{3}q^{2}(q^{5})^{2} \big) \psi + \\
\big( - q^{3}_x \lambda q^{5} + q^{5}_x \lambda q^{3} -\lambda^{2} q^{1}
q^{5} +\lambda^{2} q^{3} q^{5} q^{6} -\lambda^{2} q^{4} (q^{5})^{2}\big)%
\psi_x \\
+\big( - q^{5}_x -\lambda q^{3} q^{5} -\lambda q^{5} q^{6}\big)\psi_{2x} +
q^{5} \psi_{3x} = 0
\end{multline*}%
The substitution $\psi =\exp \int rdx$ yields a nonlinear ordinary
differential equation on the function $r$ and its first, second and third
order derivatives. This function $r$ plays the role of a generating function
of conservation law densities with respect to the parameter $\lambda $ for
the system \eqref{six}. The expansion of $r$ at infinity (i.e. $\lambda
\rightarrow \infty $)%
\begin{equation*}
r=\lambda h_{-1}+h_{0}+\frac{h_{1}}{\lambda }+\frac{h_{2}}{\lambda ^{2}}+...,
\end{equation*}%
in the above equation leads to a sequence of differential relationships between
the coefficients $h_{-1}$, $h_{0}$, $h_{1}$,\dots\ The leading term (the
coefficient of $\lambda ^{3}$) coincides with the characteristic equation of
the eigenvalues of the matrix in \eqref{eq:139}. Thus, starting from
$h_{-1}=u^k$, $k=1$, $2$, $3$, the expansion of $r$ with respect to $\lambda $
has three branches of conservation law densities, that we denote by $h_{ik}$,
$i=0$, $1$,\dots Such densities are quasihomogeneous polynomials of degrees
$\deg h_{ik}=i+1$ with respect to the grading $\deg u=0$,
$\deg \partial _{x}^{{}}=1$, and their coefficients are expressible via
rational functions of $(u^{k})$.

Using Reduce \cite{KVV17} we found all expressions of $h_{ik}$, for $k=1$,
$2$, $3$ and $i=0$: they are of the form $h_{0k} =c_{ki}(u)u_{x}^i$, where
\begin{multline}  \label{eq:6}
h_{01} = \frac{1}{S_1}\Big( - 4 u^4_x (u^{5})^{2} + u^5_x((u^{6})^{2}-2
u^{6} u^{1} + (u^{1})^{2} - (u^{2})^{2} \\
+2 u^{2}u^{3} - (u^{3})^{2}) +2 u^6_x u^{5}(- u^{6} + u^{1}) +2 u^1_x
u^{5}(-2u^{1} + u^{2} + u^{3}) \\
+2 u^2_x u^{5}(u^{2}-u^{3}) +2 u^3_x u^{5}(-u^{2} + u^{3})\Big)
\end{multline}
\begin{multline}  \label{eq:5}
h_{02} = \frac{1}{S_2}\Big(4u^4_x(u^{5})^{2} + u^5_x ( - (u^{6})^{2} + 2
u^{6}u^{2} + (u^{1})^{2} -2u^{1}u^{3} \\
- (u^{2})^{2} + (u^{3})^{2} ) + 2 u^6_x u^{5} (u^{6} - u^{2}) + 2 u^1_x
u^{5}( - u^{1}+ u^{3}) + \\
2 u^2_x u^{5}( - u^{1}+ 2 u^{2}-u^{3})+ 2u^3_x u^{5}(u^{1}-u^{3})\Big)
\end{multline}

\begin{multline}  \label{eq:7}
h_{03} = \frac{1}{S_3}\Big(-4 u^4_x (u^{5})^{2} + u^5_x ((u^{6})^{2} - 2
u^{6} u^{3} - (u^{1})^{2} +2 u^{1}u^{2} \\
- (u^{2})^{2} + (u^{3})^{2}) + 2 u^6_x u^{5}( - u^{6} +u^{3} ) + 2 u^1_x
u^{5} (u^{1}-u_{2}) \\
+ 2u^2_x u^{5}( - u^{1}+ u^{2} ) + 2 u^3_x u^{5} (u^{1} +u^{2} -2 u^{3})\Big)
\end{multline}
where $S_1= 4u^{5}((u^{1})^{2} - u^{1}u^{2} - u^{1}u^{3} + u^{2}u^{3})$, $%
S_2 = 4 u^{5} (u^{1} u^{2}-u^{1} u^{3} - (u^{2})^{2} +u^{2} u^{3})$ $S_3 = 4
u^{5} (u^{1} u^{2} -u^{1} u^{3} -u^{2} u^{3} + (u^{3})^{2})$.

We observe that the above conserved densities are not independent. It holds: 
\begin{equation}  \label{eq:8}
h_{01} + h_{02} + h_{03} = 0.
\end{equation}

We get (higher) commuting flows from the above conserved densities by
the formula in potential coordinates
\begin{equation}  \label{eq:9}
b^i_t = \tilde{g}^{ik}\partial_x^{-1}\frac{\delta h}{\delta b^k} = \tilde{g}%
^{ik}\left( - \frac{\partial h}{\partial b^k_x} + \partial_x \frac{\partial h%
}{\partial b^k_{xx}}\right) = w^i_j(\mathbf{b}_x) b^j_{xx}.
\end{equation}
In particular, from the independent densities $h_{01}$ and $h_{02}$ we obtain
the commuting flows $w^i_{1j}(\mathbf{b}_x) b^j_{xx}$ and $w^i_{2j}(%
\mathbf{b}_x) b^j_{xx}$. We observe that they are commuting flows for the
system \eqref{eq:271} in potential coordinates, and they define higher
commuting flows $(w^i_{\alpha j}u^j_x)_x$ for the system~\eqref{eq:12}.  So,
they fulfill the conditions \eqref{eq:42} and \eqref{eq:43} as they are
invariant conditions. However, we need to check the Hamiltonian property of the
operator $B$ in coordinates $(q^i)$.

In principle, it is possible to invert the coordinate change \eqref{eq:13}
and express the commuting flows $(w^i_{\alpha j}u^j_x)_x$ in coordinates
$(q^i)$.

\begin{lemma}
  The change of coordinate formula for the flow $(w^i_{\alpha j}(%
  \mathbf{u})u^j_x)_x$ into the flow $(w^i_{\alpha j}(\mathbf{q}%
  )q^j_x)_x$ is
\begin{equation*}
\frac{\partial q^k}{\partial u^h}w^h_i(\mathbf{u})\frac{\partial u^i}{%
\partial q^j} = w^k_j(\mathbf{q}).
\end{equation*}
\end{lemma}

We remain with the computational problem of expressing $(u^i)$ in terms of $%
(q^k)$. This would lead to complicated expressions involving roots, so we will
write down the matrix $w_{\alpha kj}$ for the two flows ($a=1,2$) using the
coordinates $(u^i)$ as parameters for $(q^k)$. The two matrices $w_{\alpha kj}$
turn out to be skew-symmetric, so that the condition~\eqref{eq:29} is
satisfied. All coefficients $w_{\alpha kj}$ have denominators that contain
factors of
\begin{equation}
  \label{eq:15}
  \Delta = \sqrt{\vert\det(g_{ij})\vert} = (u^1-u^2)(u^1-u^3)(u^2-u^3)u^5;
\end{equation}
if we introduce the notation $\tilde{w}_{\alpha kj} = \Delta w_{\alpha kj}$, then
the only nonzero components of the two matrices $\tilde{w}_{\alpha kj}$ (for
$k<j$) are
\begin{align}  \label{eq:16}
  & \tilde{w}_{113} =
    \tfrac{-(u^{2} + u^{3} - u^{1}-u^{6}) (u^{2} - u^{3}) u^{5}}{2}
\\
& \tilde{w}_{114} = (u^{2}-u^{3})(u^{5})^{2} \\
&\tilde{w}_{125} = \tfrac{-(u^{2}+u^{3}-u^{1}-u^{6})(u^{2}-u^{3})u^{5}}{2} \\
& \tilde{w}_{126} = (u^{2}-u^{3})(u^{5})^{2} \\
& \tilde{w}_{134} = (u^{2}-u^{3})(u^{5})^{2} u^{1} \\
& \tilde{w}_{135} = \tfrac{-(4 u^{4} u^{5} + (u^{6})^{2}-2 u^{6} u^{1} +(u^{1})^{2}
- (u^{2})^{2} +2 u^{2} u^{3} - (u^{3})^{2})
(u^{6}+u^{1}-u^{2}-u^{3})(u^{2}-u^{3})}{8} \\
& \tilde{w}_{145} = \tfrac{-(4 u^{4} u^{5}+ (u^{6})^{2} -4 u^{6} u^{1}- (u^{1})^{2}
+2 u^{1} u^{2} +2 u^{1} u^{3}- (u^{2})^{2}+2 u^{2} u^{3}- (u^{3})^{2}
)(u^{2} -u^{3})u^{5}}{4} \\
& \tilde{w}_{146} = (u^{2}-u^{3}) (u^{5})^{2} u^{1} \\
& \tilde{w}_{156} = \tfrac{-(4 u^{4} u^{5}+ (u^{6})^{2}-2 u^{6} u^{1} + (u^{1})^{2}-
(u^{2})^{2}+2 u^{2} u^{3} - (u^{3})^{2}) (u^{6} +u^{1} -u^{2} -u^{3}) (u^{2}
-u^{3})}{8}
\end{align}
and 
\begin{align}  \label{eq:18}
& \tilde{w}_{213} = \tfrac{-(u^{6} -u^{1}+u^{2}-u^{3})(u^{1} -u^{3}) u^{5}}{2} \\
& \tilde{w}_{214} = - (u^{1}-u^{3}) (u^{5})^{2} \\
& \tilde{w}_{225} = \tfrac{-(u^{6} -u^{1} +u^{2}-u^{3})(u^{1}-u^{3} )u^{5}}{ 2} \\
& \tilde{w}_{226} = - (u^{1} -u^{3} )(u^{5})^{2} \\
& \tilde{w}_{234} = - (u^{1}-u^{3} ) (u^{5})^{2} u^{2} \\
& \tilde{w}_{235} = \tfrac{(4 u^{4} u^{5} + (u^{6})^{2} -2 u^{6} u^{2} - (u^{1})^{2}
+2 u^{1} u^{3} + (u^{2})^{2}- (u^{3})^{2} ) (u^{6}-u^{1} +u^{2}-u^{3})
(u^{1} -u^{3})}{8} \\
& \tilde{w}_{245} =\tfrac{-((u^{2} -u^{3})^{2}+ (u^{1})^{2} -2(u^{2}+u^{3} ) u^{1} -
(u^{6} -4 u^{2})u^{6}-4 u^{4} u^{5})(u^{1}-u^{3})u^{5}}{4} \\
& \tilde{w}_{246} = -(u^{1} -u^{3}) (u^{5})^{2} u^{2} \\
& \tilde{w}_{256} = \tfrac{(4 u^{4} u^{5} + (u^{6})^{2} -2 u^{6} u^{2}- (u^{1})^{2}
+2 u^{1} u^{3}+ (u^{2})^{2} - (u^{3})^{2} ) (u^{6} - u^{1} + u^{2} - u^{3})
(u^{1}-u^{3}) }{8}
\end{align}

We end this paper by exhibiting the third-order Hamiltonian operator for the
Oriented Associativity equation in hydrodynamic form.

\begin{theorem}
The Oriented Associativity equation in hydrodynamic form \eqref{six} admits
the non-local third-order Dubrovin--Novikov Hamiltonian operator 
\begin{multline}  \label{eq:3}
B = \partial_x^{}\Big(g^{ij}\partial_x^{} + c^{ij}_ku^k_x + c^1 w^i_{1
k}u^k_x\partial_x^{-1}w^j_{1 h}u^h_x \\
+ c^2(w^i_{1 k}u^k_x\partial_x^{-1}w^j_{2 h}u^h_x + w^i_{2
k}u^k_x\partial_x^{-1}w^j_{1 h}u^h_x) + c^3 w^i_{2
k}u^k_x\partial_x^{-1}w^j_{2 h}u^h_x \Big)\partial_x^{}
\end{multline}
where $g^{ij}$ is the inverse of the Monge metric \eqref{eq:10}, $c^{ij}_k$
are defined through $g^{ij}$, $w^i_{1j}$ and $w^i_{2j}$ are the matrices of
the two commuting flows that we found above and 
\begin{equation}  \label{eq:4}
c^1 = 2,\quad c^2 = 1,\quad c^3 = 2.
\end{equation}
\end{theorem}

\begin{proof}
It is only necessary to check the conditions 
\begin{align}  \label{eq:20}
  \begin{split}
  &c_{ijk,l} + g^{pq}c_{pjk}c_{qil} +
    \\
   &\hphantom{c_{ijk,l} +}c^1w_{1jk}w_{1il} + c^2(w_{1jk}w_{2il} +
    w_{1il}w_{2jk}) + c^3 w_{2jk}w_{2il} = 0,
  \end{split}
     \\
  &w_{aij,s}\frac{\partial u^s}{\partial q^l} - g^{pq}c_{pij}w_{aql} =
   0,\qquad a=1,2
\end{align}
which yield the Hamiltonian property with the given values of the constants
$c^i$.
\end{proof}

\section*{Acknowledgments}

We thank Matteo Casati, Boris Dubrovin, Evgeny Ferapontov and Paolo Lorenzoni
for stimulating discussions.

This work was partially supported by the grant of Presidium of RAS
\textquotedblleft Fundamental Problems of Nonlinear Dynamics\textquotedblright\
and by the RFBR grant 17-01-00366.  RV acknowledges the support of Dipartimento
di Matematica e Fisica ``E. De Giorgi'' of the Universit\`a del Salento, of
Istituto Nazionale di Fisica Nucleare by by IS-CSN4 Mathematical Methods of
Nonlinear Physics, of GNFM of Istituto Nazionale di Alta Matematica\\
\url{http://www.altamatematica.it}.

The authors thank the Centro Internazionale per la Ricerca Matematica (Trento,
Italy) for its hospitality in March 2018 within the framework of a `research
in pairs' program, where this paper was completed.


\begin{thebibliography}{99}
\small
\bibitem{AL16} \emph{A. Arsie, P. Lorenzoni}, $F$-manifolds, multi-flat
  structures and Painlev\'e transcendents, \texttt{arXiv:1685937}.

\bibitem{CFPV18} \emph{M. Casati, E.V. Ferapontov, M.V. Pavlov and
    R.F. Vitolo}, On a class of third-order non-local Hamiltonian operators,
  J. Geom.\ Phys.\ (2018),
  \url{https://doi.org/10.1016/j.geomphys.2018.10.018},
  \url{https://arxiv.org/abs/1805.00746}.

\bibitem{Dubr} \emph{B.A. Dubrovin},  Geometry of 2D topological
field theories, Lecture Notes in Math. 1620, Springer-Verlag (1996)
120--348.

\bibitem{DN83} \emph{B.A. Dubrovin and S.P. Novikov,}  Hamiltonian
formalism of one-dimensional systems of hydrodynamic type and the
Bogolyubov-Whitham averaging method, Soviet Math. Dokl., \textbf{27} (1983)
781--785.

\bibitem{DN84} \emph{B.A. Dubrovin and S.P. Novikov}, Poisson brackets of
  hydrodynamic type, Soviet Math.\ Dokl.\ \textbf{30} No.\ 3 (1984), 651-654.

\bibitem{Fer+curv} \emph{E.V. Ferapontov,}  Nonlocal Hamiltonian
operators of hydrodynamic type: differential geometry and applications,
Amer. Math. Soc. Transl. (2), \textbf{170} (1995) 33--58.

\bibitem{FerFirst} \emph{E.V. Ferapontov,} Differential geometry of nonlocal
  Hamiltonian operators of hydrodynamic type, Func. Anal. Appl., \textbf{25}
  No. 3 (1991) 37--49.

\bibitem{FGMN97} \emph{E.V. Ferapontov, C.A.P. Galvao, O.I. Mokhov
    and Y. Nutku},
   Bi-Hamiltonian structure of equations of associativity in 2-d
topological field theory, Comm. in Math. Phys., \textbf{186} (1997) 649-669.

\bibitem{FM96} \emph{E.V. Ferapontov and O.I. Mokhov}, On the Hamiltonian
  representation of the associativity equations. In: Algebraic aspects of
  integrable systems: In memory of Irene Dorfman. Eds. I.M. Gelfand,
  A.S. Fokas, Birkh\"auser, Boston, 1996, 75--91.

\bibitem{FPV14}
\emph{E.V. Ferapontov, M.V. Pavlov and R.F. Vitolo}
 Projective-geometric aspects of homogeneous third-order hamiltonian
  operators.
 {\em J. Geom. Phys.}, 85:16--28, 2014.
 \texttt{DOI:10.1016/j.geomphys.2014.05.027}.

\bibitem{FPV16}
\emph{E.V. Ferapontov, M.V. Pavlov and R.F. Vitolo},
 Towards the classification of homogeneous third-order {H}amiltonian
  operators.
 {\em Int. Math. Res. Not.}, 22:6829--6855, 2016.

\bibitem{FPV17:_system_cl}
\emph{E.V. Ferapontov, M.V. Pavlov and R.F. Vitolo},
 Systems of conservation laws with third-order {H}amiltonian
  structures.
 {\em Lett. Math. Phys.}, 108(6):1525--1550, 2018.
 \url{https://arxiv.org/abs/1703.06173}.

\bibitem{HM} \emph{C. Hertling and Y. Manin},  Weak Frobenius
manifolds, Internat. Math. Res. Notices \textbf{1999}, no. 6, 277--286.

\bibitem{KerstenKrasilshchikVerbovetsky:HOpC}
\emph{P.~Kersten, I.~Krasil'shchik and A.~Verbovetsky},
 Hamiltonian operators and $\ell^*$-coverings.
 {\em J. Geom. Phys.}, 50:273--302, 2004.

\bibitem{KoMa07b} \emph{B. Konopelchenko and F. Magri, }\newblock Coisotropic
deformations of associative algebras and dispersionless integrable
hierarchies. Comm. Math. Phys. \textbf{274} No. 3 (2007) 627--658.

\bibitem{KoOr09} \emph{B.G. Konopelchenko and G. Ortenzi, }\newblock Coisotropic
deformations of algebraic varieties and integrable systems. J. Phys. A:
Math. Theor. \textbf{42} No. 41 (2009) 415207, 18 pp.

\bibitem{KVV17}
\emph{J.~Krasil'shchik, A.~Verbovetsky and R.~Vitolo},
 {\em The symbolic computation of integrability structures for partial
  differential equations}.
 Texts and Monographs in Symbolic Computation. Springer, 2018.
 ISBN 978-3-319-71654-1; see \url{http://gdeq.org/Symbolic_Book} for
  downloading program files that are discussed in the book.

\bibitem{LPR11} \emph{P. Lorenzoni, M. Pedroni and A. Raimondo},
  $F$-manifolds and integrable systems of hydrodynamic type, Archivum
Mathematicum \textbf{47} (2011), 163-180.

\bibitem{Lor14} \emph{P. Lorenzoni}, Darboux--Egorov System, Bi-flat
  $F$-Manifolds and Painlev\'e VI, IMRN \textbf{12} (2014), 3279--3302.

\bibitem{Magri} \emph{F. Magri},  A simple model of the integrable
Hamiltonian system, J. Math. Phys., \textbf{19} No. 5 (1978) 1156--1162.

\bibitem{FM90:_nonloc_hamil}
\emph{O.I. Mokhov and E.V. Ferapontov},
Non-local {H}amiltonian operators of hydrodynamic type related
  to metrics of constant curvature, Russ.\ Math.\ Surv.\ no.\ 3 \textbf{45}
  (1990), 218--219.

\bibitem{SPori} \emph{M. Pavlov and A. Sergyeyev}, Oriented associativity
  equations and symmetry consistent conjugate curvilinear coordinate nets,
  J. Geom.\ Phys.\ \textbf{85} (2014), 46--59.

\bibitem{PV15} \emph{M.V. Pavlov and R.F. Vitolo,} On the bi-Hamiltonian
  geometry of WDVV equations, Lett.\ Math.\ Phys.\ \textbf{105}, no. 8 (2015),
  1135-1163.

\bibitem{VitoloHO} \emph{R. Vitolo}, Computing with Hamiltonian Operators,
  \url{https://arxiv.org/abs/1808.03902}.

\end{thebibliography}
\end{document}